\documentclass[11pt]{article}

\usepackage{graphicx}

 \usepackage{epsfig}
\usepackage{amssymb}
\usepackage{graphicx}
\usepackage{multirow}
\usepackage{subcaption}
\usepackage{amsmath}
\usepackage{xspace}
\usepackage{tabularx}
\usepackage[sort]{cite}
\usepackage{fullpage}
\usepackage{enumitem}
\usepackage{algorithm}
\usepackage{amsthm}
\usepackage{hyperref}
\usepackage{framed}
\usepackage{color}
\usepackage[noend]{algpseudocode}

\newtheorem{theorem}{Theorem}

\newtheorem{lemma}[theorem]{Lemma}
\newtheorem{observation}[theorem]{Observation}

\newcommand{\knw}{${\sf KNOWN}$\xspace}
\newcommand{\kwn}{${\sf KNOWN}$\xspace}
\newcommand{\unkwn}{${\sf UNKNOWN}$\xspace}

\newcommand{\pingpong}{{\sc Ping-Pong}\xspace}
\newcommand{\pingpongk}{{\sc K-Ping-Pong}\xspace}
\newcommand{\placeswipe}{{\sc Place-\&-Swipe}\xspace}
\newcommand{\aLook}{\textsf{Look}\xspace}
\newcommand{\aMove}{\textsf{Move}\xspace}
\newcommand{\aCompute}{\textsf{Compute}\xspace}

\title{Patrolling on Dynamic Ring Networks}

\author{Shantanu Das$^{\dag}$, Giuseppe A. Di Luna$^{\dag}$, Leszek A. Gasieniec$^{\ddag}$\\
\footnotesize
$^\dag$ Aix-Marseille University, LIS, CNRS, Universit\'{e} de Toulon, France \\
\footnotesize
{\{shantanu.das,giuseppe.diluna\}@lis-lab.fr} \\
\footnotesize
	$^\ddag$ University of Liverpool, Liverpool, UK\\
	\footnotesize
	{L.A.Gasieniec@liverpool.ac.uk}}
\date{}

\begin{document}
\maketitle

\begin{abstract}
We study the problem of patrolling the nodes of a network collaboratively by a team of mobile agents, such that each node of the network is visited by at least one agent once in every $I(n)$ time units, with the objective of minimizing the idle time $I(n)$. While patrolling has been studied previously for static networks, we investigate the problem on dynamic networks with a fixed set of nodes, but dynamic edges. In particular, we consider 1-interval-connected ring networks and provide various patrolling algorithms for such networks, for $k=2$ or $k>2$ agents. We also show almost matching lower bounds that hold even for the best starting configurations. Thus, our algorithms achieve close to optimal idle time. Further, we show a clear separation in terms of idle time, for agents that have prior knowledge of the dynamic networks compared to agents that do not have such knowledge. 
This paper provides the first known results for collaborative patrolling on dynamic graphs.  
\end{abstract}

\section{Introduction}
In recent years patrolling is gaining 
on popularity in the area of algorithms and in particular algorithmics of mobile agents and applications.
Patrolling naturally occurs in daily routines requiring regular visits to specific (possibly mobile) objects and areas.  
It can also refer to monitoring of complex network processes or systems behaviour. 
Typical applications of patrolling include safety or security related surveillance, regular updates, data gathering, and other perpetual tasks.

We consider the patrolling problem in networks (graphs) with the objective of visiting all nodes of the graph perpetually, optimizing the \emph{idle time} - the maximum time period during which any node is left unvisited. Unlike all previous results on the patrolling problem, we study the problem on a dynamic graphs where some links of the graph may be missing for certain duration of time. This complicates the problem and requires a strong coordination between the agents, in order to reduce the idle time, even in simple networks. We restrict our attention, in this paper to dynamic ring networks. In the case of a static ring network, the simple strategy of periodically cycling the nodes of the ring, is known to provide the optimal idle time. However, for patrolling dynamic rings, more involved strategies are required depending on the number of the agents, the capabilities of the agents and whether or not the dynamic structure of the network is known to the agents. Among various known dynamic graph models, we consider interval connected dynamic networks which ensures that the network is connected at any time interval. We distinguish between the \kwn setting when the agents know in advance about the changes in the graph structure, from the \unkwn setting when such information is not available to the agents. We show a clear separation between the two cases, in terms of the minimum idle time for patrolling. For both cases, we provide lower bounds and almost matching upper bounds on the idle time for patrolling, supported by deterministic algorithms for collaborative patrolling.

\subsection*{Related Work}
\paragraph{Patrolling.}
The problem of patrolling is a close relative to several classical algorithmic challenges which focus on monitoring and mobility. 
These challenges include the Art Gallery Problem \cite{DBLP:journals/ipl/Ntafos86}, 
where one is interested in determining the smallest number of   
inert guards and their location to constantly monitor 
all artefacts, and its dynamic alternative referred to as 
the k-Watchmen Problem 
\cite{DBLP:journals/ipl/Chin88,DBLP:journals/dcg/CarlssonJN99}. 
In further work on fence patrolling \cite{DBLP:conf/esa/CzyzowiczGKK11, DBLP:journals/dc/KawamuraK15,CollinsCGKKKMP13} the authors focus on monitoring vital (possibly disconnected) parts of a linear environment where each point is expected to be visited with the same frequency. A similar approach is adopted in \cite{DBLP:journals/algorithmica/CzyzowiczGKKKT17} where we find studies on monitoring of a linear environment by agents prone to faults. 
The problem of patrolling objects which require different frequencies of visits was first considered in \cite{DBLP:conf/sofsem/GasieniecKLLMR17}, 
where the authors assume availability of a single mobile agent.
They also showed a close relationship between these type of patrolling and the Pinwheel scheduling problem \cite{DBLP:journals/algorithmica/ChanC93}. 
In a more recent work \cite{DBLP:conf/sofsem/GasieniecKLLMR17} the authors consider monitoring by two agents of $n$ nodes located on a line and requiring different frequencies of visits. The authors provide several approximation algorithms concluding with the best currently known $\sqrt{3}$-approximation.

\paragraph{Dynamic networks and mobile agents.}
The field of dynamic networks is an hot and active research topic \cite{HarG97,KuO11,CaFQS12,Michail16}. In the message passing model a lot of attention has been devoted to classic problems such as agreement \cite{KuMO11,AuPR13,BRS12}, information dissemination \cite{AwE84,KLyO10,OdW05,CleMPS11}, and counting \cite{di2015brief,icalpcounting}. Surprisingly, the investigation of mobile agents on dynamic networks started only recently. In the centralised setting (when agents know the dynamic of the graph apriori) the problem of exploring a graph in the fastest possible way has been studied in several papers \cite{Michail2014,ErHK15,cover1}. The task is NP-hard on general graphs and it becomes polynomial on special topologies \cite{cover2,IlKW14}. Notably, in the case of interval connected ring the exploration can be done in ${\cal O}(n)$ rounds \cite{IlcinkasW18}. 

The distributed setting (when agents do not know the dynamic of the graph) has been mostly overlooked, or limited to restrictive dynamic assumptions, such as periodic \cite{Flocchini2013,IlW11} or recurrent \cite{IlcinkasW18} graphs. The exploration with termination of interval connected rings has been studied in \cite{DBLP:conf/icdcs/LunaDFS16}. For rings that are connected over time, a perpetual self-stabilizing exploration algorithm has been proposed in \cite{BournatDD16}. Finally, the gathering problem on  interval connected rings has been studied in \cite{LunaFPPSV17}.
To the best of our knowledge there is no previous work studying the patrolling of a dynamic network. 

\subsection*{Our Contributions}

\begin{table}[tb]
\center
\resizebox{0.6\textwidth}{!}{%
\begin{tabular}{ l  l | l | l |}
\hline
Adversary&&\multicolumn{2}{ |c| }{Number of Agents} \\
\cline{3-4}
& & $k=2$ & $k>2$ \\ \hline
 \multirow{2}{*}{\kwn}&  Upper Bound & $3\lceil \frac{n}{2} \rceil$& $3\lceil \frac{n}{k} \rceil$\\
 \cline{2-4}
  &  Lower Bound & $n$ & $\frac{2n}{k}$ \\
\hline

 \multirow{2}{*}{\unkwn}&  Upper Bound  & $2n-2$& $4 \lceil \frac{n}{k} \rceil$\\
  \cline{2-4}
  &  Lower Bound & $2n-6$& $\frac{2n}{k}$ \\
\hline
\end{tabular}}
\vspace{0.2cm}
\caption{ Results for the idle time in dynamic rings of $n$ nodes, with $k$ uniformly placed agents having global visibility. \label{table:iammiamm}}
\end{table}

We show, first of all, in Section~\ref{sec:local}, that when the agents have local visibility, limited to the current node, then patrolling has an idle time of $n-\alpha \cdot k$ rounds, both in case of arbitrary initial placement (where $\alpha=1$) and uniform initial placement with $b$-bits of persistent memory (where $\alpha=2^{b}$). This means that using multiple agents reduces the idle time by only an additive factor. In contrast, for a \emph{static} ring, the idle time for patrolling with $k$ agents is $\frac{n}{k}$, achieving a multiplicative factor efficiency over single agent patrolling. 

Thus, for the rest of paper, we consider agents having global visibility, allowing it to see the current configuration of the ring with set of available links.
We start with team size of $k=2$ agents in Section~\ref{global:two} and then generalize these results to $k>2$ agents in Section~\ref{global:k}. The results of these two sections are summarized in Table~\ref{table:iammiamm}. The bounds denoted here are for the stable idle time, after a stabilization time that is at most $O(n)$. These results show a clear distinction between the case of \kwn adversary (where the dynamic structure of the network is known apriori) and the case of \unkwn adversary when the agents do not have prior knowledge of the dynamic network.

Other than the above results, we also show a slightly better lower bound of $\lfloor (1+\frac{1}{5})(n-1) \rfloor$ for the special case of two agents in dynamic ring with \kwn adversary, when the agents are arbitrarily placed.

\section{Model}
A set of agents, $A:\{a_0,\ldots,a_{k-1}\}$, operates on a dynamic graph ${\cal G}$. Each agent follows the same algorithm (all agents are identical) executing a sequence of \aLook,\aCompute, \aMove cycles. 
In the \aLook phase of each cycle, the agent acquires a {\em snapshot} of the environment. In the {\aCompute} phase the agent uses the information from the snapshot and the contents of its local persistent memory to compute the next destination, which may be the current node or one of its neighbors. During the \aMove phase an agent traverses an edge to reach the destination node. The information contained in the persistent memory is the only thing that is preserved among cycles.

\paragraph{\bf Synchronous system.}
The system is {\em synchronous}, that is the time is divided in discrete units called rounds. Rounds are univocally mapped to numbers in $\mathbb{N}$, starting from round $0$. In each round, each agent in $A$ executes exactly one entire \aLook,\aCompute, \aMove cycle. 

\paragraph{\bf Interval connected ring.}
A dynamic graph ${\cal G}$ is a function mapping a round $r \in \mathbb{N}$ to a graph $G_r:(V,E(r))$ where $V:\{v_0,\ldots,v_{n-1}\}$ is a set of nodes and $E: \mathbb{N} \rightarrow V \times V$ is a function mapping a round $r$ to a set of undirected edges. 
%
We restrict ourselves to $1$-interval-connected rings. A dynamic graph ${\cal G}$ is a $1$-interval-connected ring when the union of the graph instances $G_{\infty}=(V,E_{\infty})=(V,\cup_{i=0}^{+\infty} E(i))$ is a ring graph,  
and at each round $r$, the graph $G_r$ is connected. 
The graph ${\cal G}$ is anonymous, i.e. all nodes are identical to the agents. The endpoints of each edge are labelled 
as either {\em clockwise} or {\em counter-lockwise}, in a consistent manner (i.e the ring is oriented). 

\paragraph{\bf Local versus Global Snapshot.}
\begin{itemize}
\item Local Snapshot: the snapshot obtained by an agent at a node $v$ in round $r$ contains only information about the node $v$, i.e. the number of  agents in $v$ and the set of available edges incident to node $v$ at round $r$.
\item Global Snapshot: the snapshot obtained by an agent contains the graph $G_r$ (where the current location of the agent is marked), and for each node in $V$ the number of agents present in that node at round $r$.  
\end{itemize}


\paragraph{\bf Knowledge of ${\cal G}$.}

We examine two different settings: the one with known ${\cal G}$ (\knw) and the one without such knowledge (\unkwn). 
In the \knw setting during the {\aCompute} phase agents have access to the dynamic graph ${\cal G}$. In this case the decision of what will be the movement of the agent depends on the snapshot,  on the content of the persistent local memory and on the entire dynamic graph ${\cal G}$. On the contrary in the \unkwn setting, during the  { \aCompute} phase, no other information is available (an agent uses only the snapshot and the local memory). Another way to see the \unkwn setting is to imagine that ${\cal G}$ is adaptive to the strategy of algorithm ${\cal A}$: there exists an adversarial entity, namely the {\em scheduler}, that decides the graph ${\cal G}$ according to the strategy of algorithm ${\cal A}$. 

\paragraph{\bf Configurations and initial placement of agents.}
Given a graph $G_r$, and the set of agents $A$, a configuration at round $r$, is a function $C_r: A \rightarrow V$ that maps agents in $A$ to nodes of $V$ where agents are located. 
We say that there is a {\em uniform initial placement}, if $C_0$ is such that the segments of consecutive rings nodes not occupied by agents have size  $\lfloor \frac{n}{k} \rfloor$ or $\lceil \frac{n}{k} \rceil$. We say that there is an {\em arbitrary initial placement} if the configuration $C_0$ is injective ( no two agents may start on the same node).

\paragraph{\bf Idle time.}
An algorithm ${\cal A}$ running on a graph ${\cal G}$, generates an execution ${\cal E}$. The execution ${\cal E}:\{C_0,C_1,C_2\ldots\}$ is an infinite sequence of configurations, one for each round $r$. 
Given a node $v$ and an execution ${\cal E}$ the set of visits of $v$, $S_{{\cal E},v}:\{r_1,r_2,r_3,\ldots \}$ is a set containing all rounds in which $v$ has been visited by some agent in execution ${\cal E}$; more formally, $r_j \in S_{{\cal E}, v}$ if and only if $C_{r_j}(a)=v$ for some $a \in A$. The idle set $I_{{\cal E},v}$ of node $v$ is a set containing all the intervals of time between two consecutive visits of node $v$ in execution ${\cal E}$; more formally, $x \in I_{{\cal E},v}$ if and only if there exists $r_i,r_{i-1}$ in $S_{{\cal E}, v}$ and $x=r_i-r_{i-1}$. We assume that each node has been visited at round $-1$.

We say that an algorithm solves patrolling on a graph ${\cal G}$, if each node of the graph is visited infinitely often.  
Given an algorithm ${\cal A}$ and an integer $n \geq 5$, we define as $T_n$ the set of all executions of algorithm ${\cal A}$ over any ($1$-interval-connected) dynamic ring ${\cal G}$ with $n$ nodes. The idle time of algorithm ${\cal A}$ is  the function $I(n)=\underset{{\forall {\cal E} \in T_n}}{\max}( \cup_{\forall v \in V}I_{{\cal E},v})$.

\paragraph{Stable idle time.}
Given an execution ${\cal E}$ we define as ${\cal E}[r,\infty]$ the execution obtained by removing the first $r$ configurations from ${\cal E}$, similarly we can define the idle set  $I_{{\cal E}[r,\infty],v}$ . An algorithm ${\cal A}$ as a stable idle time $I_{r_s}(n)$ with stabilisation time $r_s$ if there exists a round $r_s$ such that $I_{r_s}(n)=\underset{{\forall {\cal E} \in T_n}}{\max}( \cup_{\forall v \in V}I_{{\cal E}[r_s,\infty],v})$.

\section{Preliminaries}

We devote this section to some simple observations based on previous results on dynamic rings. Note that for a single agent moving in a dynamic ring, an adaptive adversary can keep the agent confined to the starting node and one of its neighbors.

\begin{observation} (\cite{KLyO10,OdW05}) \label{obs:2node}
In a dynamic ring ${\cal G}$ under the \unkwn model with global snapshot, 
a single agent can visit at most $2$ nodes. 
\end{observation}

\begin{observation} (\cite{KLyO10,OdW05}) 
In a dynamic ring ${\cal G}$ under the \knw model, a single agent can reach any node $V$ in at most $n-1$ rounds. 
\end{observation}

Due to the above observations, the only interesting cases for patrolling is for $k\geq 2$ which we investigate in this paper. For any $k$ agents, we have the following observation derived from the proof of Proposition 1 in \cite{IlcinkasW18}.

\begin{observation}\label{obs:ilcinkas} (\cite{IlcinkasW18}) 
Given a dynamic ring ${\cal G}$ the \unkwn model with Global Snapshot. For any round $r$ and any $1 \leq h \leq n-1$, there are $n-h$ distinct nodes, such that if $n-h$ agents are placed in these nodes and they all move in the same direction from round $r$ until round $r+h-1$, then they visit exactly $h+1$ nodes. 
\end{observation}

\begin{proof} The proof is contained in \cite{IlcinkasW18}. We report it here for completeness. 
Let us imagine to have an agent on each node. At each round an agent move counter-clockwise (or clockwise). 
W.l.o.g round $r=0$.
The proof is by induction:
\begin{itemize}
\item (Base Step) Let $h=1$. Round  is $r=0$. There is at most one edge missing, at most on agent is blocked. Thus there are $n-1$ agents that visit $2$ nodes. 
\item (Inductive step )Let $h=t+1$. Round is $r=t$. From the inductive hypothesis we that $n-t$ agents visited $t+1$ nodes by round $t-1$. At round $t$ at most $1$ of this agent is blocked, thus we have $n-t-1$ agents that visited $t+2$ nodes. 
\end{itemize}\end{proof}

It is also possible to show an easy lower bound on the idle time of any algorithm under the strongest model considered in this paper (i.e. under global visibility and knowledge of ${\cal G}$) 

\begin{theorem}\label{2knw:lb}
Consider the \knw model with Global Snapshot.
Let ${\cal A}$ be any patrolling algorithm for $k$ agents with uniform initial placement. We have that $I_{r_s}(n) \geq \frac{2n}{k}$ for any stabilization time $r_s$. 
\end{theorem}
\begin{proof}
The scheduler removes the same edge forever. At this point the $k$ agents have to patrol a line and the lower bound for idle time on a line with $k$ agents  is $\frac{2n}{k}$ (See \cite{CollinsCGKKKMP13} for a proof).  
\end{proof}

\section{Patrolling with Local Visibility}
\label{sec:local}

In this section we analyse the Local Snapshot model, we first examine the case in which the placement of the agents is arbitrary and then we examine the case in which the placement is uniform.

\subsection{Lower bound for arbitrary initial placement}

\begin{theorem}\label{th:trivial}
Consider a dynamic ring under the \knw model with Local Snapshot and arbitrary initial placement.
Then any patrolling algorithm ${\cal A}$ for $k$ agents has stable idle time $I_{r_s}(n) \geq n-k$, for any stabilisation time $r_s$. 
\end{theorem}

\begin{proof}
Let us consider a static ring of $n$ nodes $G=(V=\{v_0,\ldots,v_{n-1}\}, E=\{(v_0,v_1),(v_1,v_2),\ldots\} )$ and a set of agents $\{a_0,\ldots,a_{k-1}\}$. Configuration $C_0$ is such that $C(a_j)=v_{j}$, that is agents are placed one on each node in $\{v_0,\ldots,v_{k-1}\}$. As the nodes of the ring are anonymous and the agents are identical with local visibility, each executing the same algorithm, at each round $r$ the configuration $C_r$ can only be a rotation of configuration $C_0$. 
Moreover, configuration $C_r$ is a rotation of either one step counter-clockwise or one step clockwise of configuration $C_{r-1}$. This implies that the best idle time is obtained by having agents to perpetually move in the same direction.
The idle time of this strategy is $I_{r_s}(n) = n-k$ for any possible stabilization time $r_s$. 
\end{proof}

The above result assumes the agents to be placed on consecutive nodes, and its proof does not hold when there is an uniform initial placement of agents.  
We consider the case of uniform placement in the next section.
%
%

\subsection{Lower bound for uniform placement in the \unkwn model} 

We now prove a lower bound on the idle time for any patrolling algorithm for $k$ agents with uniform initial placement in dynamic rings under the \unkwn model. This result holds only for agents with bounded memory.

\begin{theorem}\label{theorem:lowerboundlocal}
Consider a dynamic ring under the \unkwn model with local snapshots and uniform initial placement.
Given any patrolling algorithm ${\cal A}$ for $k$ agents, with $c={\cal O}(1)$ bits of memory, the idle time for patrolling is $I(n) \geq n-7\cdot2^{c}k$. 
\end{theorem}

In order to prove this result we have to introduce some concepts related to the state diagram of a patrolling algorithm ${\cal A}$

\paragraph{State diagram.}
Given an algorithm ${\cal A}$ executed by an agent $a_j$, we can model it as a finite state machine with state diagram $H_{\cal A}$. We use the terms vertex and arc when we refer to the state diagram to no generate confusion with the terms edge and node used for the dynamic ring.  We may also use the term state when we refer to a vertex of $H_{\cal A}$.

Let $T_{\cal A}:(S,D)$ be the projection of $H_{\cal A}$ obtained consider only the arcs and vertices of  $H_{\cal A}$ that are visited in executions where agents never meet. 

Each vertex $s \in S$ in $T_{\cal A}:(S,D)$ has three outgoing arcs, and each arch has a label in the form $Snapshot:Movement$. One arc leads to the state that is reached when the agent sees that both edges are incident in the local node (this arc has a label with $Snapshot=\{C,CC\}$). The second arc is the one corresponding to the state transition that agent does when the missing edge is the counter-clockwise one (this arc has a label with $Snapshot=\{C\}$). Finally, the last arc is the one used when the missing edge is the clockwise one ($Snapshot=\{CC\}$). 

Each arc label has also associated the movement $m$ that the agent performs when in state $s$ it sees a specific local snapshot (let us recall that a snapshot corresponds to the label of the arc). We have that $m \in \{0,-1,1\}$, $m=0$ if the agent stays at the current node, it is $-1$ if the agent moves to the counter-clockwise, and it is $+1$ if the agent moves to the clockwise.  
In Figure~\ref{fig:simple} there is the $T_{\cal A}$ of a simple algorithm in which an agent goes in fixed direction until it is blocked by an edge removal. Once blocked the agent switches direction.

\begin{figure}

  \centering
    \includegraphics[width=0.6\textwidth]{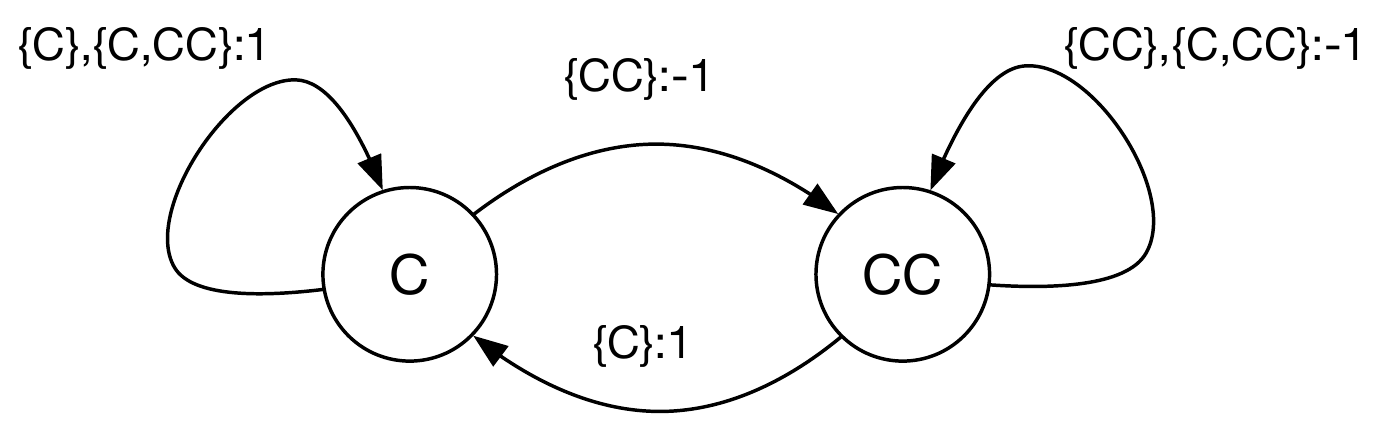}
      \caption{Simple algorithm in which agents reverse direction when blocked. The arcs have label in the form $Snapshot:Movement$. The Snapshot is $\{C,CC\}$ if both edges are present and $\{C\}$ (resp. $\{CC\}$) if the clockwise (resp. counter-clockwise) edge is absent. The movement is $1$ if the agents move on the clockwise edge, $-1$ if it moves on the counter-clockwise edge and $0$ if it stays still. }\label{fig:simple}
\end{figure}

\paragraph{Reachability, fault free paths and cycle displacement.}
Given two vertices $s,s' \in S$ we say that, in $T_{\cal A}$, vertex $s'$ is reachable from $s$ in $t$-steps if and only if there exits a simple directed path in $T_{\cal A}$ from $s$ to $s'$ that has length $t$. Notice that, the presence of such a path means that starting from state $s$ there always exists a scheduler of edge removals that forces the agent to reach $s'$, and in such scheduler the agent will traverses at most $t$ edges of the ring. Also notice that we must have $t \leq |S|^2$. 

Given a path $p$ in $T_{\cal A}$ we say that $p$ is a fault-free path, if any arc in $p$ has label with $Snapshot=\{C,CC\}$. The definition of fault-free cycle is analogous. 
Given a cycle $F$ in $T_{\cal A}$ the displacement of $d(F)$ is the sum of all the movements on the arcs in $F$. Essentially, a cycle has zero displacement when an agent placed at node $v$ at the end of cycle is still in node $v$. 
A cycle has positive displacement if, after the executions of all actions in the cycle, the agents moved clockwise. The negative displacement is analogous. 
See Figure~\ref{fig:cycles}  for examples of previous definitions. 

\begin{figure}

  \centering
    \includegraphics[width=0.4\textwidth]{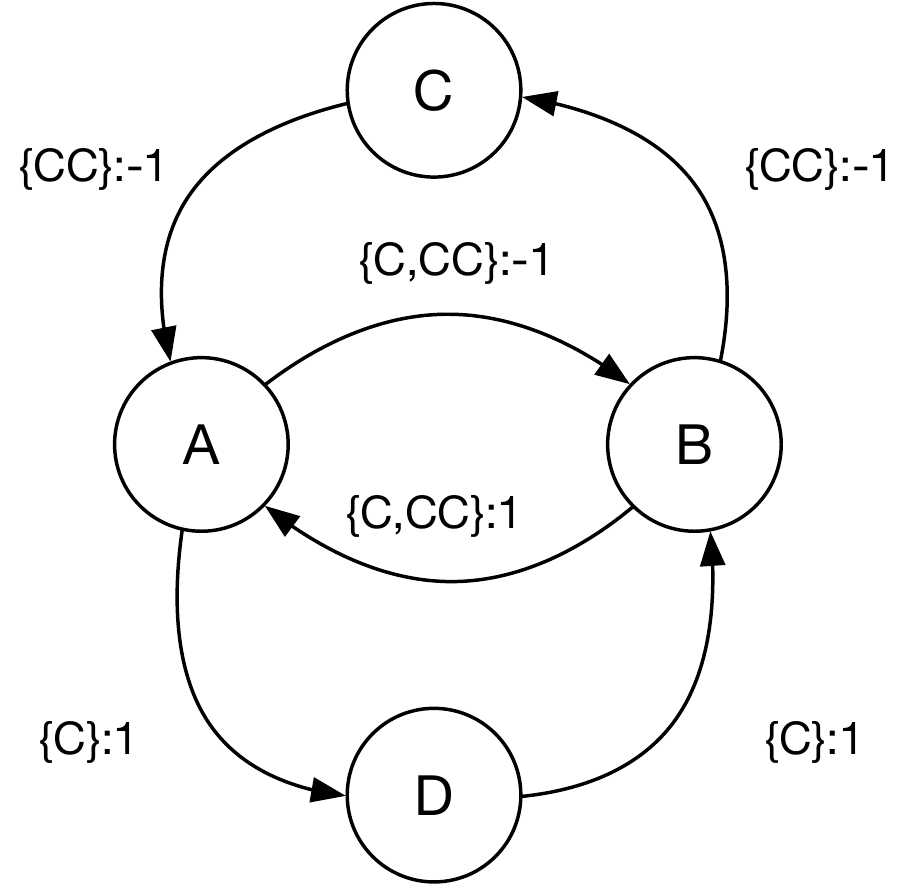}
      \caption{Example of cycles: we have a fault-free cycle $(A,B),(B,A)$ with zero displacement, a cycle $(B,C),(C,A),(A,B)$ with negative displacement equal to $-3$, and a cycle $(A,D),(D,B),(B,A)$ with positive displacement equal to $3$.}\label{fig:cycles}
\end{figure}

Before entering in the proof of our lower bound, we do a simple observation on the structure of any algorithm ${\cal A}$.

\begin{observation}\label{obs:trivial}
Let us consider any algorithm ${\cal A}$. Consider any state $s$ of $T_{\cal A}$. There always exists a path $p$, possibly empty, from $s$ to a vertex $s'$ such that:
\begin{itemize}
\item $p$ is fault-free.
\item $s'$ belongs to a cycle $F_{s}$ of $T_{\cal A}$, and the cycle $F_{s}$ is fault free.
\end{itemize}
\end{observation}
\begin{proof}
The proof comes directly from the fact that there is no sink vertex in graph $U=(S,D[\{C,CC\}])$, with $D[\{C,CC\}]$ subset of $D$ containing only arcs with label $\{C,CC\}$. 
\end{proof}

From Observation~\ref{obs:trivial} we have that for any algorithm ${\cal A}$ with initial state $s_0$ the cycle $F_{s_0}$ is well defined and exists. 

The proof of our lower bound is based on the following two lemmas. 

\begin{lemma}\label{lemma:nopatrol}
Let us consider any algorithm ${\cal A}$ with initial state $s_0$. If from a vertex $s \in F_{s_0}$ there exists a path in $T_{\cal A}=(S,D)$ to a node $s^{*}$ such that $s^{*}$ belongs to a strongly connected component $K$ of $T_{\cal A}$ and to a fault-free cycle $F^{*}$ with $d(F^{*}) = 0$, then ${\cal A}$ does not solve patrolling.  
\end{lemma}
\begin{proof}
Let us consider a ring of size $n=4k^2(|S|+|F_{s_0}|+|F^{*}|)$. Now we create a scheduler that prevents ${\cal A}$ from patrolling. We first present the scheduler, then will the prove that such scheduler does not make agents to meet, and thus the state of each agent remains in $T_{\cal A}$.

The scheduler first takes agent $a_{0}$ and forces it to go in state $s^{*}$, notice that this takes at most $|S|$ rounds. Thus, from round $r_1=|S|$ agent $a_{1}$ moves perpetually over a set of nodes of size at most $|F^{*}|$.

Now the scheduler waits until $a_1$ enters in state $s$, this must happen: by  Obs~\ref{obs:trivial} eventually the internal state of $a_1$ is in $F_{s_0}$. 
Once agent $a_{1}$ is in state $s$ the adversary  forces it to go to $s^{*}$. This takes at most $|S|+|F_{s_0}|$ rounds, counting also the number of rounds needed by $a_1$ to enter in state $s$. 
By using this scheduler we have that at round $r_2=2|S|+|F_{s_0}|$ agents $a_{0},a_{1}$ are both in cycle $F^{*}$. This means that they perpetually move over a set of nodes of size at most $|F^{*}|$.

By iterating the previous process for the agents  $a_{2},a_{3},\ldots,a_{k-1}$ we have that at round $r_{k}=k|S|+(k-1)|F_{s_0}|$ all agents are in cycle $|F^{*}|$. From this round on, all agents perpetually over a set of nodes that has size at most $k|F^{*}|$.

The set of nodes visited by all agents is upper bounded by $r_{k}+k|F^{*}|$, that is $k|F^{*}|+k|S|+(k-1)|F_{s_0}|$.
It is immediate to see that half of the nodes in the ring have not been explored, an they will not be explored at any point in the future.

It remains to show that during this process two agents do no meet. The maximum amount of node traversed by a single agent is upper bounded by $r_{k}+|F^{*}| \leq 2r_{k}$. However, the initial space between two agents is at least $4r_{k}$, recall that $4r_{k} \leq \frac{n}{k}$. Thus no two agents meet. 
This complete the proof. 
 \end{proof}

\begin{lemma}\label{lemma:dance}
Let us consider any algorithm ${\cal A}$ with initial state $s_0$. If from a node $s \in F_{s_0}$ there exists a path in $T_{\cal A}=(S,D)$ to a node $s^{*}$  such that $s^{*}$ belongs to a strongly connected component $K$ and a fault-free cycle $F^{*}$ with $d(F^{*}) \neq 0$, then the patrolling time of  ${\cal A}$ is $I(n) \geq n-7|S|k$. 
\end{lemma}
\begin{proof}

Let us consider a ring of size $n=4k^2(|S|+|F_{s_0}|+|F^{*}|)$. We create a scheduler that forces a maximum distance of $n-f(|S|)k$ between two agents,  and that does not allow two agents to meet.  

The scheduler uses two phases, in the first phase it forces each agent to enter in cycle $F^{*}$, this is done in an analogous way to the one used in the proof of Lemma~\ref{lemma:nopatrol}, and we will omit its description. In the second phase the scheduler reduces the distances among agents until the maximum distance between two agents is $n-7|S|k$. 

The scheduler first takes $a_{k-2}$ and it blocks it on two neighbour nodes until the distance between $a_{k-2}$ and $a_{k-1}$ is $7|K|$. At this point the adversary it forces $a_{k-2}$ to go to $s^{*}$, this can be done since $K$ is strongly connected.
Notice that, at the end of such process the distance between $a_{k-2}$ and $a_{k-1}$, decreased by at most $2|K|$ rounds. This means that their distance is at least $5|K|$. We now show that
such distance is big enough to ensure that, if $a_{k-2}$ and $a_{k-1}$  are never blocked they never meet while they both cycle in $F^{*}$. 

The two agents are executing the same cycle $F^{*}$ but they are not synchronized, that is they are in two different vertices of the cycle. To prove that they never meet it is sufficient to show that they
do not meet in a period of $2F^{*}$ rounds. 
First of all, notice that during the execution of $F^{*}$ each of them moves of at most $|F^{*}| \leq |K|$ edges. Second notice that if $a_{k-1}$ terminates the cycle $F^{*}$ at round $r'$,
moving to the counter-clockwise of $d(F^{*}) \leq |K|$ edges, then by round $r'+|F^{*}|-1$ also $a_{k-2}$ moved to the counter-clockwise of  $d(F^{*}) \leq |K|$ edges. 

This ensures that at each round the distance between 
the two is at least $2|K|$ edges and at most $7|K|$ edges.

Now the scheduler uses the same procedure for agent $a_{k-3}$, putting it at distance at most $7|K|$ from $a_{k-2}$.
This procedure is iterated until agent $a_{0}$.
At this point, the maximum distance between $a_0$ and $a_{k-1}$ is $n-7k|K|$. 
 \end{proof}

From the two previous lemmas we can prove theorem~\ref{theorem:lowerboundlocal}.
Given  ${\cal A}$  with initial state $s_0$, we have to show that there exists a vertex $s \in F_{s_0}$ such that from $s$ there is a path to a strongly component $K$ of $T_{\cal A}$. But this is immediate consequence of the fact that ${\cal A}$ uses finite memory and that there is no sink vertex in $T_{\cal A}$.
From Observation~\ref{obs:trivial} we have that from any vertex in component $K$ we can reach a fault free cycle $F^{*}$. Lemma~\ref{lemma:nopatrol} ensures that the displacement of $F^{*}$ is not zero. Therefore we can use Lemma~\ref{lemma:dance} proving the claim of the theorem. 

\section{Two agents with Global Visibility}
\label{global:two}

In this section we assume that the agents have access to a global snapshot of the configuration at each round during \aLook phase.
We first consider the simpler case of $k=2$ agents and show upper and lower bounds on patrolling for both the \unkwn and the \knw setting. 

\subsection{\unkwn setting}

\begin{figure}[h]
  \centering
    \includegraphics[width=0.7\textwidth]{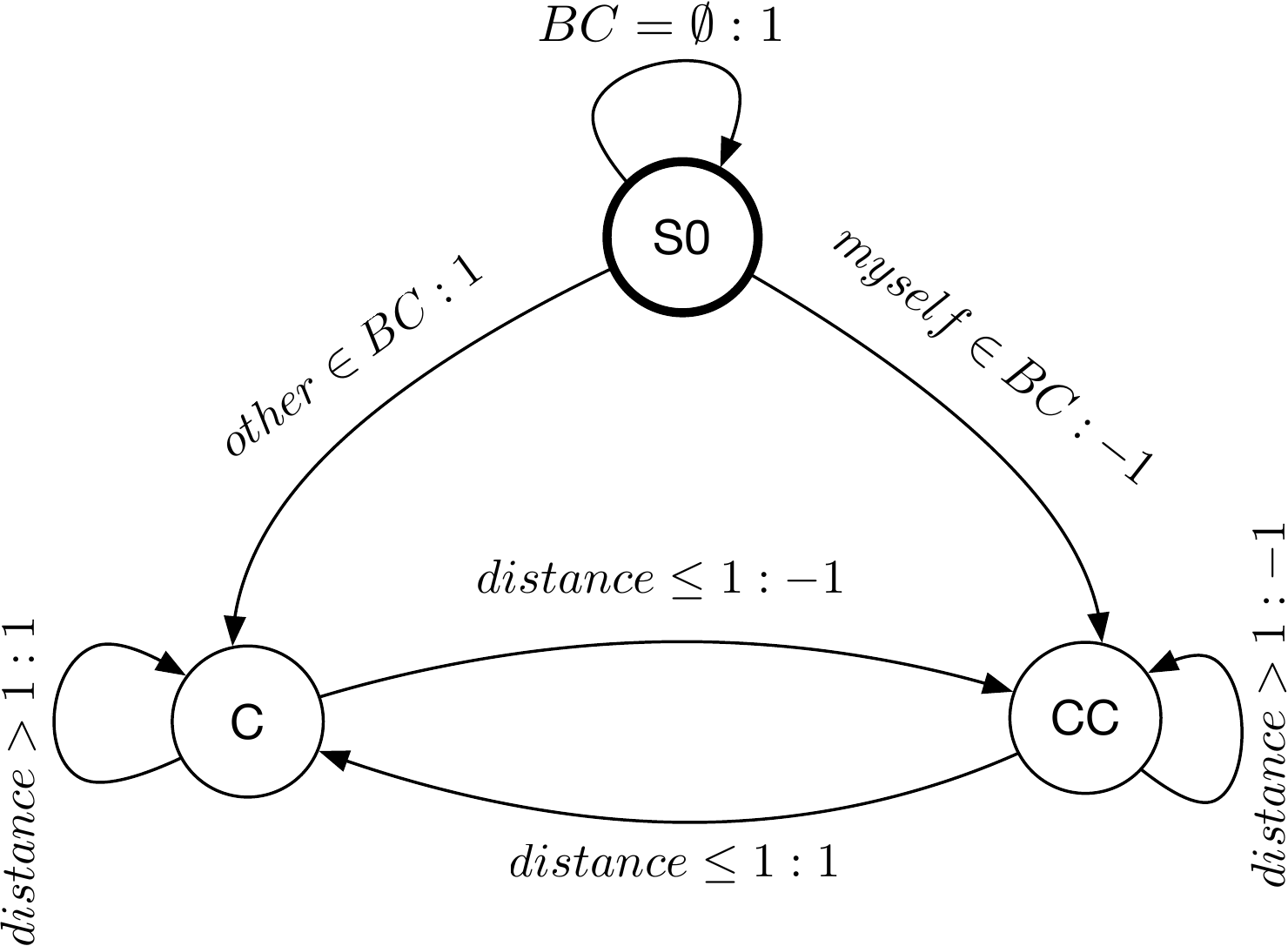}
      \caption{Algorithm  \pingpong  state diagram. The starting state is S0. Transition are of the form $Predicate:Movement$ where values of $1,-1,0$ denotes clockwise, counter-clockwise or no move, respectively.
      }\label{fig:pingpong}
\end{figure}

Given the graph $G_r$ at round $r$, we define as $BC_{r}$ (resp. $BCC_{r}$) the set of all agents that are attempting to move clockwise (resp. counter-clockwise) from a node $v$ that has the clockwise (resp. counter-clockwise) edge missing at the round $r$. 
We will remove the subscript $r$ when it is clear that we are referring to the current round.

\smallskip

We now describe a patrolling algorithm called \pingpong for $k=2$ agents in the \unkwn setting.
Initially, both agents move in the clockwise direction in each round, until they reach a round $r$ in which $BC_{r}$ is not empty. 
At this point the symmetry between agents is broken, and we assign to the agent in $BC_{r}$ the counter-clockwise direction while the other agent keeps  the clockwise direction. Starting from round $r$, the agents continue to move according to the following rule:
Move in the assigned direction until the minimum distance between the agents is less or equal to $1$; When this happen, both agents reverse their direction (i.e, the agents bounce off each other). 
The state diagram of the algorithm is presented in Figure~\ref{fig:pingpong}. 


\begin{theorem}\label{2unkwn:pingpong}
For any dynamic ring in the \unkwn model with Global Snapshot and arbitrary initial placement, Algorithm \pingpong allows two agents to patrol the ring with an idle time $I(n) \leq 2(n-1)$.
\end{theorem}

\begin{proof}
The algorithm has two distinct phases. In the first phase, both agents move in the same direction, while in the second phase the agents always move in opposite directions. We need to show that for any node $v$, given two consecutive visits of $v$ at round $r_0$ and $r_1$ it holds that $r_1-r_0 \leq 2(n-1)$. 
First, let $r_0$ and $r_1$ be both in the first phase of the algorithm. Observe that in this phase each agent loops around the ring visiting each node once in every $n$ rounds. Since the agents on distinct nodes we have at most $n-1$ rounds between two visits of node $v$; thus $r_1-r_0 \leq n-1$.

Now we examine the case when $r_0$ and $r_1$ are both in the second phase. It takes at most $n-1$ rounds for the distance between the two agents to be $1$ or less--the agents are moving on opposing direction and at most one of them can be blocked at any round. This means that during a period that is upper bounded by $n-1$ all nodes are visited. Thus, there are at most $2(n-1)$ rounds between consecutive visits of a node $v$. 

Finally, we have to show that the bound still hold if $r_0$ is in the first phase and $r_1$ in the second. Let $r$ be the round in which the algorithm switches phase. We necessarily have $r-r_0=x \leq n-1$, by the previous discussion regarding the first phase. At round $r$, one agent is at distance $x$ from node $v$ and thus, the distance between the agents on the segment not containing $v$, is at most $(n-x-1)$. Now, if both agents are move towards $v$, then $v$ would be visited in the next $(n-1)$ rounds. Otherwise, the agents move away from $v$, therefore in at most $(n-x-2)$ rounds, the two agents would be at distance one or less. In the subsequent $n-1$ rounds all nodes would be visited (recall our previous discussion for the second phase). This implies that $r_1-r_0 \leq 2(n-1)$ in both cases.
 \end{proof}

Surprisingly, the algorithm \pingpong is almost optimal. 

\begin{theorem}\label{2unkwn:lb}
Under the \unkwn model with global snapshot and uniform initial placement, any patrolling algorithm ${\cal A}$ for two agents has idle time $I(n) \geq 2n-6$. 
\end{theorem}

\begin{proof}
We show that the adversarial scheduler can (1) entrap one of the agents on two neighboring nodes of the ring, say, nodes $v_{n-1},v_{n-2}$, and at the same (2) prevent the other agent from performing a full tour of the ring.  
Under the above two conditions, patrolling the ring by two agents reduces to patrolling a line of $l=n-2$ nodes by a single agent, for which we have an idle time of $2(l-1)=2n-6$.

Note that condition (1) can be easily achieved by the adversary (see Observation~\ref{obs:2node}). Suppose the adversary traps agent $a$ in the two nodes $v_{n-1},v_{n-2}$. We will now show that the other agent, call it agent $b$, cannot traverse the segment $S_4$ containing the $4$ consecutive nodes $(v_0, v_{n-1}, v_{n-2}, v_{n-3})$. Suppose the agent $b$ approaches this segment from $v_0$ (i.e. clockwise direction); the other direction can be symmetrically treated. 

Figure~\ref{fig:leftgate} depicts all the configurations reachable from this scenario.
The adversary can keep the agent $a$ trapped in the two nodes until the other agent reaches node $v_{n-1}$. At this time if the two agents are together, the adversary can keep both trapped using the same argument as before. The only other possibility is if the two agents are on the neighboring nodes $v_{n-1},v_{n-2}$, in which case the adversary removes edge $(v_{n-2}, v_{n-3})$, preventing any agent from leaving the segment $S_4$ in clockwise direction. Thus either both agents are trapped, or one of the agents can leave the segment by node $v_0$, i.e. the same direction in which the agent entered the segment.
Notice the agents may swap roles, so that agent $a$ leaves and agent $b$ is blocked in nodes $v_{n-1},v_{n-2}$. In either case, the condition (1) and (2) are satisfied and therefore the theorem holds.
\end{proof}

\begin{figure}
  \hspace{-0.5cm}
    \includegraphics[width=\textwidth]{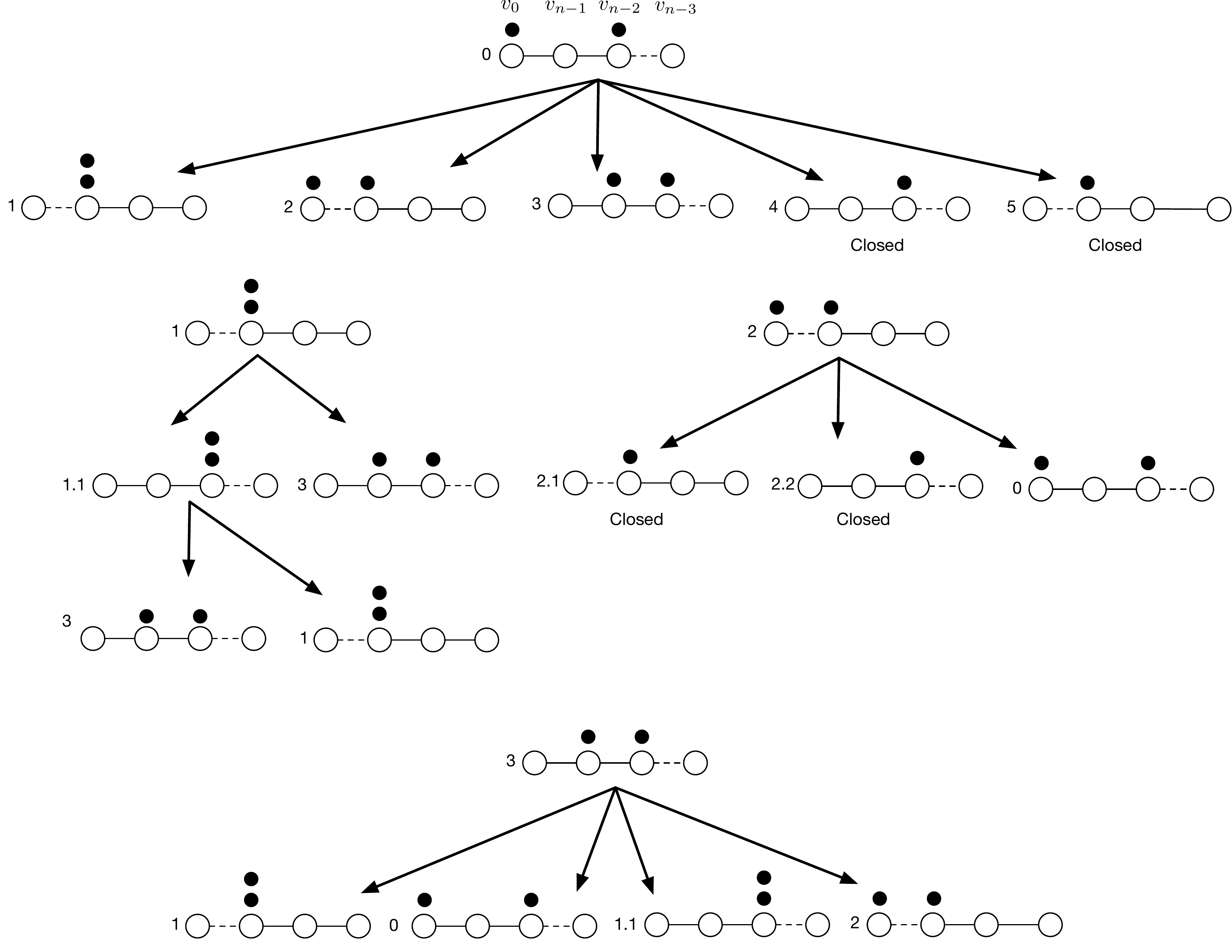}
      \caption{ Tree of reachable configurations for Theorem~\ref{2unkwn:lb}: agents are the black dot, ring nodes the white circles, and the missing edge the dotted line. We represent the set of reachable configurations as a tree, where we omit the child configuration when it is equal to the father configuration. A configuration is {\em closed} when there is only one agent in the $4$ nodes, that is also trapped in $v_{n-1},v_{n-2}$. We reach a closed configuration only when an agent moves using the counter-clockwise edge of node $v_{0}$. We can see that starting from the root, (or alternatively from Configuration $2$) all the reachable configurations are either {\em closed} configurations or configurations where no agent is on $v_{n-3}$.}\label{fig:leftgate}
\end{figure}

\subsection{\knw setting}\label{sec:placeandswipe}

In this subsection we examine the \knw setting. We first present a solution algorithm, namely \placeswipe, that solves the problem with an idle time of $3 \lceil \frac{n}{2} \rceil$ rounds, when there is an uniform initial placement of the agents.  
We then discuss how the algorithm can be adapted to work under arbitrary initial placement  by having a stabilisation time of $\lfloor \frac{n}{2} \rfloor$ and a stable idle time of $3 \lceil  \frac{n}{2} \rceil$ rounds. 
 

\begin{figure}[htb]
  \centering
    \includegraphics[width=0.8\textwidth]{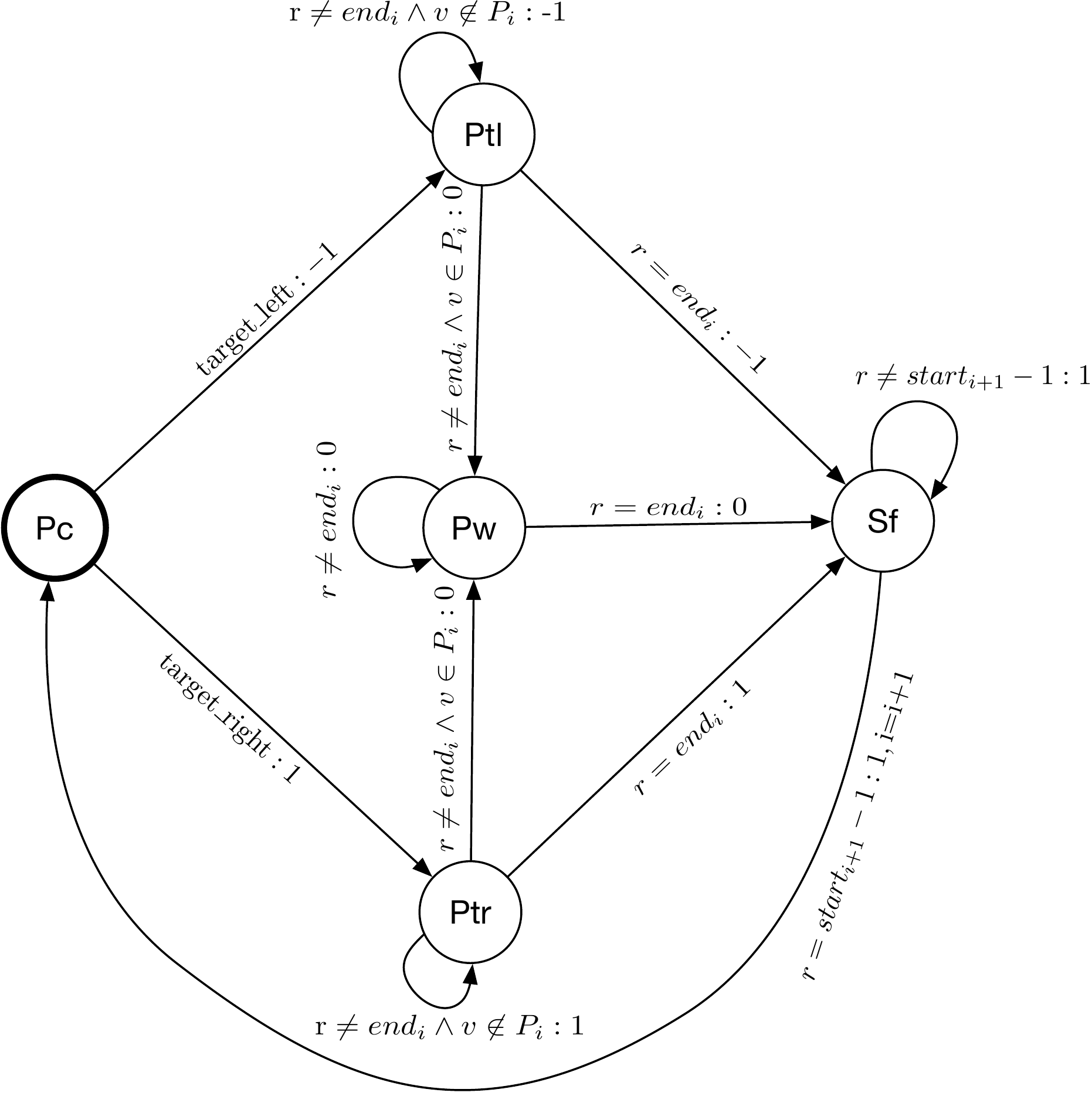}
      \caption{Algorithm  \placeswipe. The starting state in Pc, the Swiping Phase is Sf, and transitions are of the form $Predicate:Movement$}\label{fig:placeswipe}
\end{figure}

\paragraph{Patrolling Algorithm}
The algorithm \placeswipe (see Figure~\ref{fig:placeswipe}) perpetually alternates between two phases of fixed length (each phase lasts $\lceil \frac{n}{2} \rceil$ rounds). During the first phase, called \emph{Placement Phase}, the agents position themselves on a specially choosen pair of antipodal\footnote{A pair of nodes is antipodal if the  distance between them in the ring is $\lfloor \frac{n}{2} \rfloor$.} nodes -- the {\em swiping nodes}. In the second phase, called the Swipe Phase, the agents together visit all nodes of the ring by both moving clockwise for $\lfloor \frac{n}{2} \rfloor$ rounds without stop. 
A Placement Phase followed by Swipe Phase is an epoch of the algorithm, we use $i \geq 0$ to indicate the epoch number. 
Since every node is visited once in every Swipe phase, in the worst case, a node may be visited at the beginning of a Swipe phase and subsequently at the end of the next Swipe Phase, giving an idle time of at most $3 \lceil  \frac{n}{2} \rceil$ rounds.

We now show that for each epoch $i$, there exists a a special pair $P_{i}$ of antipodal nodes which allow the Swipe Phase to cover all nodes in $\lfloor \frac{n}{2} \rfloor$ rounds. Let $start_i=i  \cdot n$, and $end_i=\lceil (\frac{1}{2}+i)n \rceil -1$ be the starting and ending round of the $i$-th Placement Phase.

\begin{lemma}\label{lemma:swipe}
Given any dynamic ring ${\cal G}$ and any round $r=end_i+1$, 
there exists a pair of antipodal nodes $P_i$, such that two agents placed on $P_i$ and moving clockwise from round $end_i+1$ to the end round $start_{i+1}-1$, explore all nodes of the ring. 
\end{lemma}

\begin{proof}
The key idea to prove the existence of $P_{i}$ is Observation \ref{obs:ilcinkas}. By plugging $t=\lceil \frac{n}{2} \rceil -1$ in the statement of the observation. We have that there are $\lfloor \frac{n}{2} \rfloor +1$ nodes, let $E_i$ be this set, such that an agent being on one of these nodes at round
$end_i+1$ moving clockwise visits exactly $\lceil \frac{n}{2} \rceil$ nodes by the end of round $start_{i+1}-1$. Now we have to prove that $E_i$ contains a pair of antipodal nodes. But this is obvious since the ring contains at least $\lfloor \frac{n}{2} \rfloor$ antipodal pairs and the cardinality of $E_i$ is $\lfloor \frac{n}{2} \rfloor +1$.
Being the pair $P_{i}$ antipodals, when each agent visits $\lceil \frac{n}{2} \rceil$ nodes the ring has been explored. 
\end{proof}

%
To prove correctness of the algorithm, we need to show that agents starting from any uniform configuration, the two agents can reach the chosen nodes $P_i$ in $\lceil \frac{n}{2} \rceil$ rounds. 
Note that, for computing $P_i$ in each epoch, the algorithm needs only the knowledge of the future $n$ rounds of ${\cal G}$.

\begin{theorem}\label{th:placeswipe}
Consider the \knw model with Global snapshot and uniform initial placement. The algorithm \placeswipe allows two agents to patrol a ring with an idle time $I(n) \leq  3\lceil \frac{n}{2} \rceil$.
\end{theorem}
\begin{proof}
We first assume that agents are always able to reach $P_{i}$ in the $i$-th Placement Phase. 
By Lemma \ref{lemma:swipe} the agents explore in the Swipe Phase. Therefore, the idle time is upper bounded by the time that passes between the first round of Swipe Phase $i$ and the last round of Swipe Phase $i+1$. Since each Phase lasts $\lceil \frac{n}{2}\rceil$ rounds and between two Swipe Phases there is only one Placement Phase, then the idle time is $3\lceil \frac{n}{2} \rceil$ rounds. 
It remains to show that agents are able to reach $P_{i}$ during the $i$-th Placement Phase. 
First notice that at round $start_{i}$ agents are in antipodal positions: they start antipodal, and is easy to verify that at end of each Swipe Phase they are still antipodal.  
We have to show that given two agents on antipodal positions, they can reach, using the knowledge of ${\cal G}$, any pair of target antipodal positions in the interval $[start_{i},start_{i}+ \lceil \frac{n}{2}\rceil -1]$.
Let $a_0,a_1$ be the agents and $v_i,v_j$ the target nodes disposed as in Figure \ref{fig:antipodal}. W.l.o.g we assume $x \leq \lfloor\frac{n}{2}\rfloor -x$

\begin{figure}[H]

  \centering
    \includegraphics[width=0.4\textwidth]{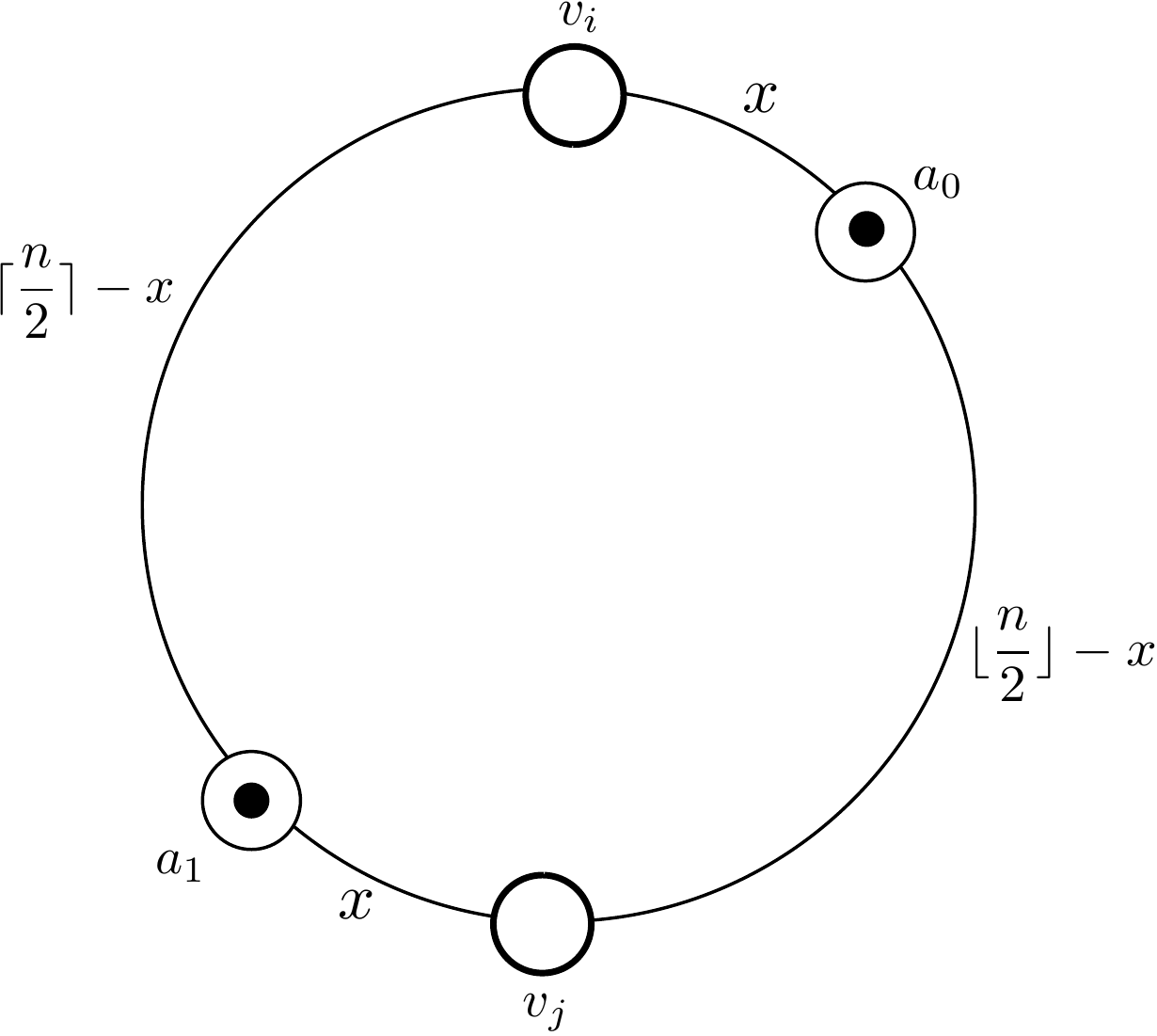}
      \caption{The values on the ring arcs are the number of edges between the depicted nodes.}\label{fig:antipodal}
\end{figure}
If agent $a_0$ reaches node $v_i$ using the counter-clockwise path $p$ and $a_1$ does the same with $v_j$ during the interval $[start_{i},start_{i}+\lceil \frac{n}{2} \rceil-1]$, then we are done. 

Otherwise, let us assume, w.l.o.g., that agent $a_0$ cannot reach node $v_i$ using the counter-clockwise path $p$ during the interval $[start_{i},start_{i}+\lceil \frac{n}{2} \rceil-1]$. Then in ${\cal G}$ agent $a_0$ is blocked for at least $\lceil \frac{n}{2} \rceil-x+1$ rounds while trying to traverse path $p$. This implies that there are at least $\lceil \frac{n}{2} \rceil-x+1$ rounds in which edges that are not on path $p$ are present in the interval $[start_{i},start_{i}+\lceil \frac{n}{2} \rceil-1]$. But this implies that, by moving in clockwise direction, agent $a_0$ reaches node $v_j$ and agent $a_1$ node $v_i$.  \end{proof}

\paragraph{Arbitrary initial placement.}
Theorem \ref{th:placeswipe} assumes that agents are starting at uniform distance.  However, it is possible to easily adapt the algorithm to work under any initial placement sacrificing the stabilization time. Essentially, we need an initialization phase in which agents place themselves in antipodal positions. This can be done in $\lfloor \frac{n}{2} \rfloor$ rounds: in each round, agents move apart from each other increasing the distance by at least one unit per round. Thus, we obtain an algorithm with stabilization time $r_s=\lfloor \frac{n}{2} \rfloor$ and $I_{r_s}(n) \leq 3\lceil \frac{n}{2} \rceil$.

\paragraph{Lower bounds.}
A lower bound of $n$ for the \knw setting is immediate from Th.~\ref{2knw:lb}. However, when the initial placement of the agents is arbitrary we can show a slightly better bound. 

\begin{theorem}\label{2knw:lbiplacement}
Let ${\cal A}$ be a patrolling algorithm for two agents with arbitrary initial placement under the \knw model with Global snapshot. For any even $n \geq 10$, there exists a $1$-interval connected ring where ${\cal A}$  has an idle time $I(n) \geq \lfloor (1+\frac{1}{5})(n-1) \rfloor$. 
\end{theorem}
 \begin{figure}[htb]
\centering
\begin{subfigure}{0.45\textwidth}
\centering
    \includegraphics[width=\linewidth]{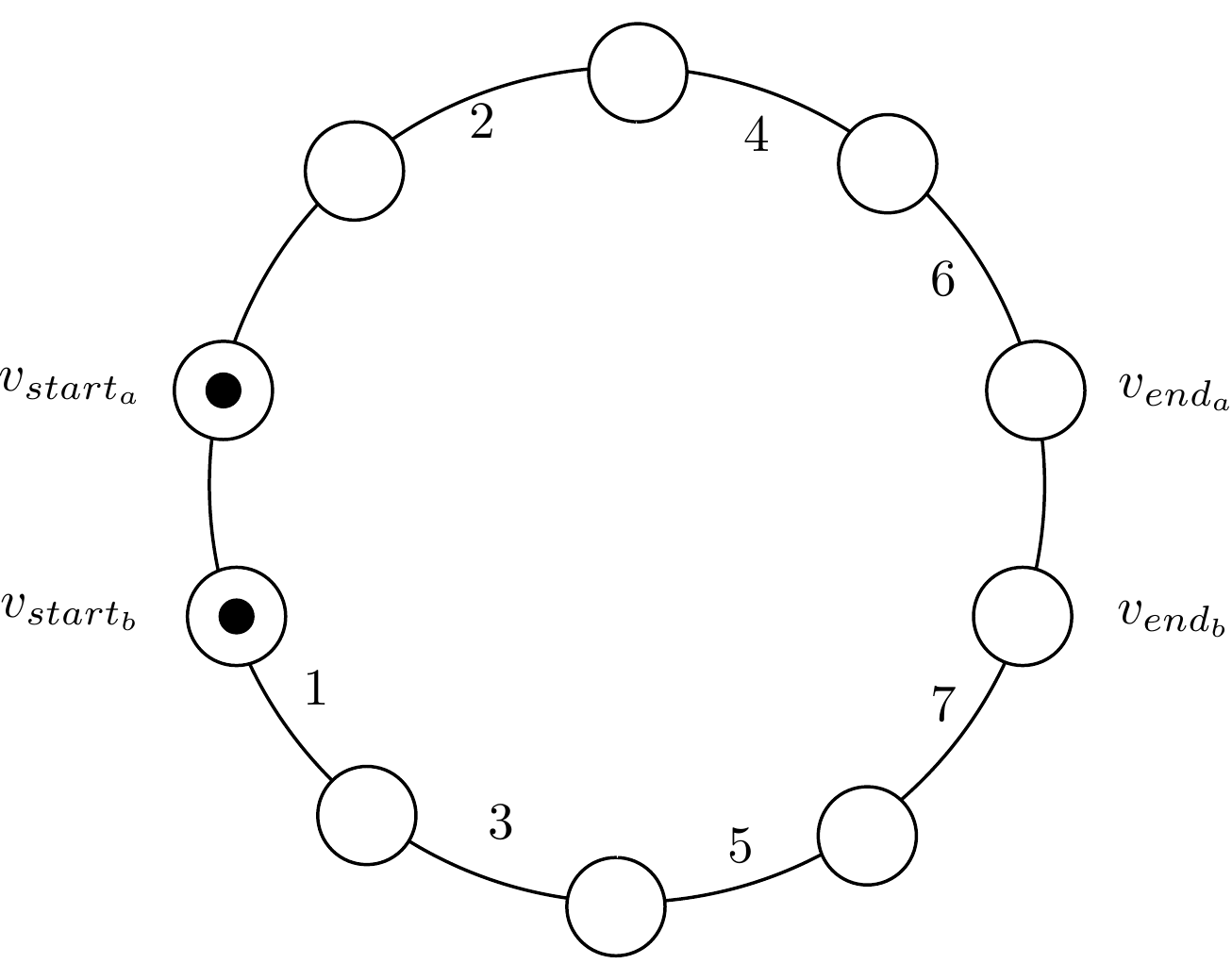}
    \caption{Example of lower bound graph: the edge label is the round in which the edge is removed in the start-end wave.}
    \label{fig:lowerbound1}
\end{subfigure}%
\quad
\begin{subfigure}{0.45\textwidth}
\centering
    \includegraphics[width=\linewidth]{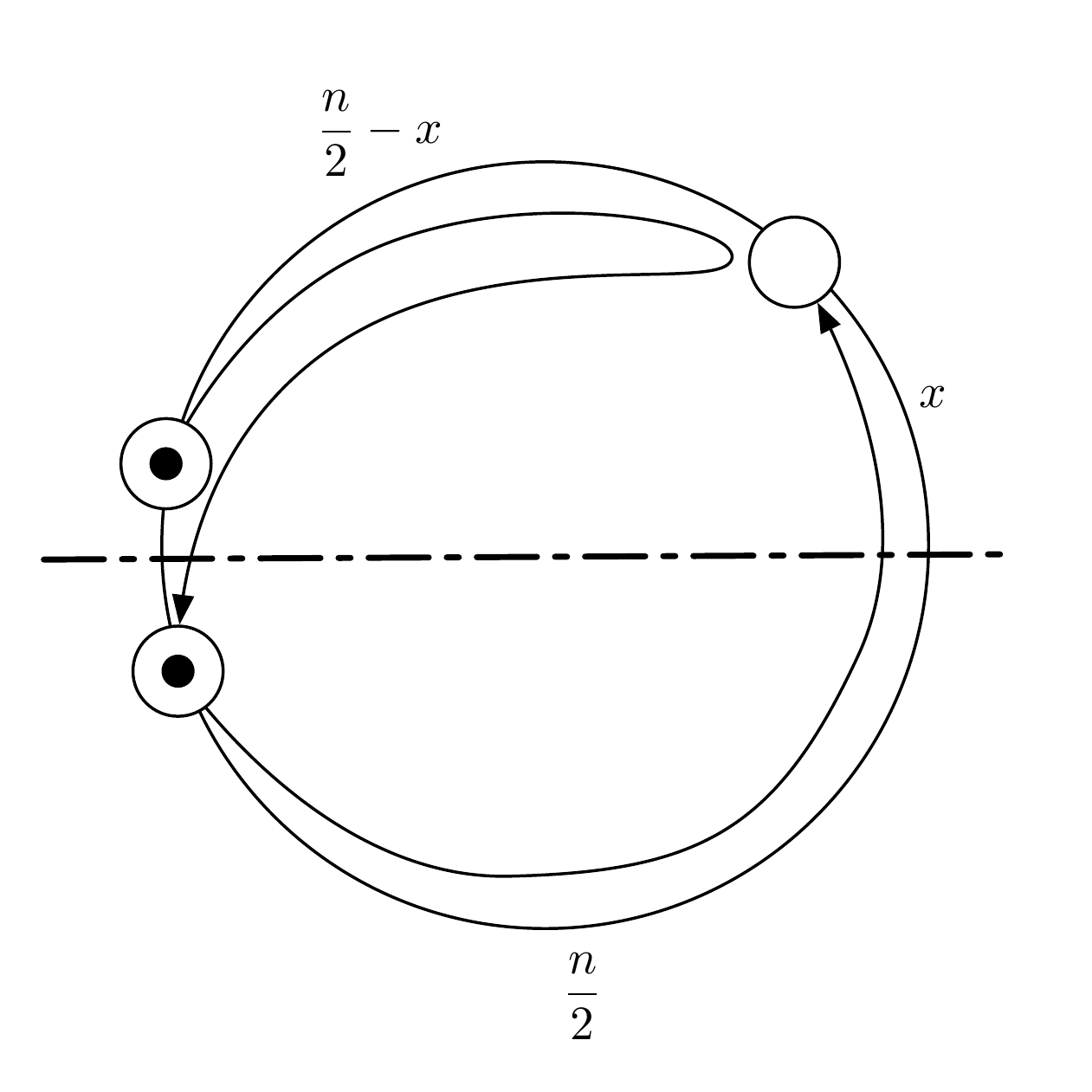}
    \caption{ Optimal Strategy for exploring all nodes an visit twice $v_{start_{a}}$.}
    \label{fig:lowerbound}
\end{subfigure}
\caption{Lower bound for the \knw model with Global Snapshot and arbitrary initial placement.}
\label{fig:problem}
\end{figure}

\begin{proof}
The agents $a,b$ are initially positioned on two neighbour nodes, let them be $v_{start_{a}},v_{start_b}$ and let $v_{end_{a}},v_{end_b}$ the antipodal nodes, see Figure~\ref{fig:lowerbound1}. The edges removal follows a {\em wave} strategy, in which in alternating periods of ${\cal O}(n)$ rounds, edges are first removed in the direction that goes from the two start nodes to the end nodes and vice versa. See Figure \ref{fig:lowerbound1} for an example of start-end wave, If an agent tries to reach the end nodes from start nodes during a start-end wave it takes at least $n-2$ rounds.

More formally, the scheduler of edges removal is as follows:
Let $e^{up}_i$ (resp $e^{down}_i$) be the edge at distance $i$ from $v_{start_{a}}$ (resp. $v_{start_b}$) in the clockwise segment $v_{start_{a}},v_{end_{b}}$ (resp. counter-clockwise segment $v_{start_{b}},v_{end_{a}}$). Note that $i \in [0,\frac{n}{2}-2]$.
Edge $e^{up}_i$ is absent in rounds $4x(\frac{n}{2}-1)+2i$ and  $4(x+1)(\frac{n}{2}-1)-2(i+1)$ for $x\in \mathbb{N}$.
Edge $e^{down}_i$ is absent in rounds $4x(\frac{n}{2}-1)+(2i+1)$ and  $4(x+1)(\frac{n}{2}-1)-(2i+1)$ for $x\in \mathbb{N}$.

Let us recall that the idle time is the maximum among the number of rounds between two consecutive visits of the same node, and the rounds needed to explore all nodes the first time. Therefore we focus on the minimum time that agents need to first explore all nodes and than to visit again node $v_{start_{a}}$.
If agents try to minimise the exploration time, then they pay at least $2n-4$ rounds to go back to $v_{start_{a}}$: they have to move in parallels to  nodes  $v_{end_{a}},v_{end_b}$, this takes $n-2$ rounds; after, at least one has to go back to $v_{start_{a}}$ and this takes other $n-2$ rounds.

Therefore, any strategies has to leave some nodes unexplored in the first $n-2$ rounds, one of the two agents has to go back to the start position before reaching the end nodes. While the other agent will take care of the nodes left unexplored by the first.

The strategy is reported in Figure \ref{fig:lowerbound}, w.l.o.g. agents $b$ covers $v_{start_b}-v_{end_b}$ and it moves of  $x$ steps in segment $v_{start_a}-v_{end_a}$. Agent $a$ moves clockwise of $d=\frac{n}{2}-1-x$ edges  and it goes back to its initial position. The best $x$ is the one such that $a$ visits again  $v_{start_a}$ exactly when $b$ completes the exploration. Thus imposing $3d=2(x+\frac{n}{2})$, by algebraic manipulation, we have $x=\frac{n}{10}-\frac{6}{10}$.
The time to explore all nodes is $2(\frac{n}{2}+x)$ (agent $b$ is blocked half of the rounds), thus the time is at least $\lfloor (1+\frac{1}{5})(n-1) \rfloor$ rounds. \end{proof}

\section{Patrolling with $k>2$ agents having Global Visibility}
\label{global:k}

In this section we examine the case of $k>2$ agents, showing how to generalize the algorithms of Section \ref{global:two} for this case. 

\subsection{\unkwn setting: Generalising  \pingpong for $k$ agents}
We generalize \pingpong for $k$ agents assuming that: $k$ divides $n$, $k$ is even, and that there is uniform initial placement. At the end of the section we discuss how to remove such assumptions.
The new algorithm, called \pingpongk (see Figure \ref{fig:pingpongk}) is divided in two phases, {\em Single-Group-Swiping}  and {\em Two-Groups-Swiping}, as described below.
\begin{figure}[htb]
  \centering
    \includegraphics[width=\textwidth]{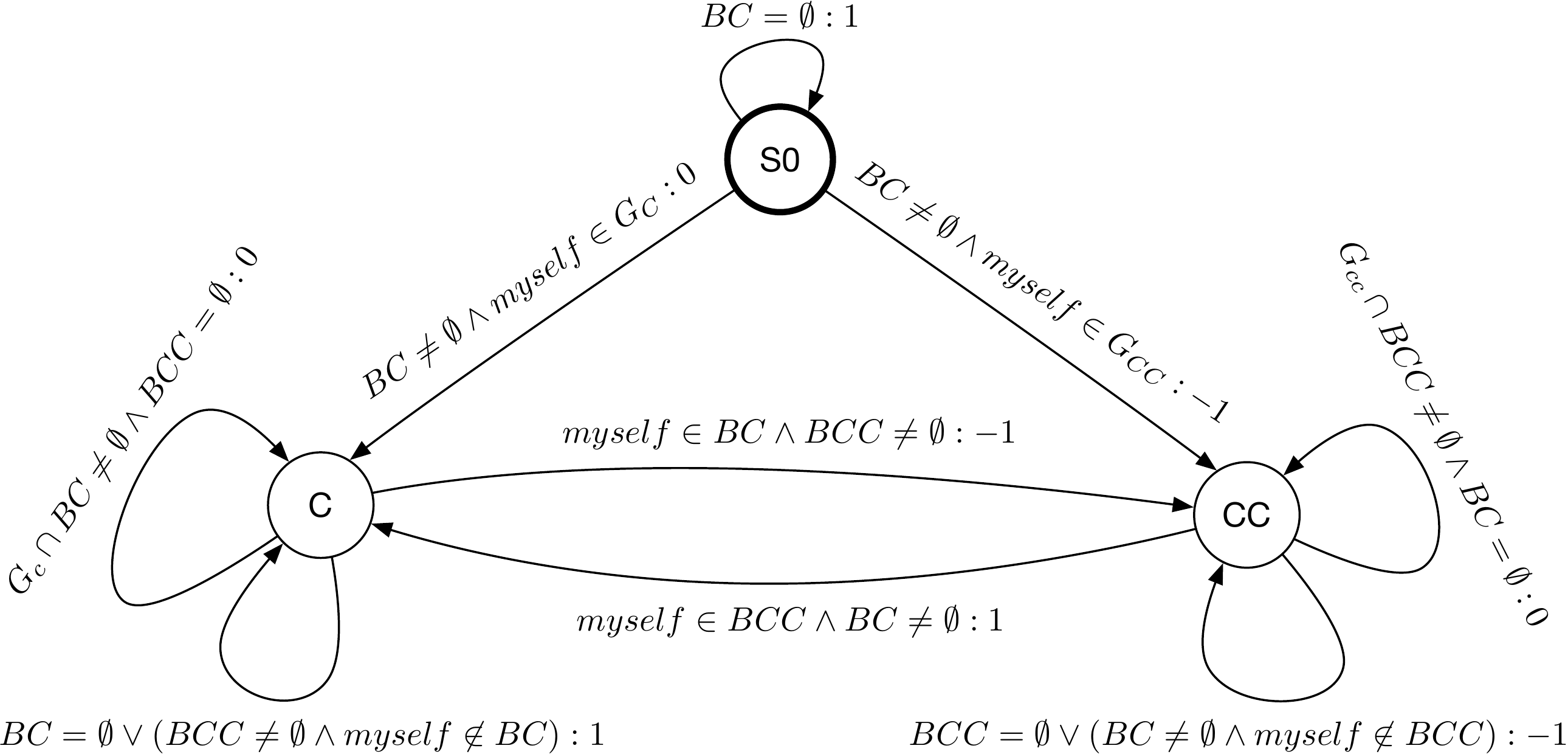}
      \caption{Algorithm  \pingpongk  state diagram. The starting state is S0. Movement on clockwise or counter-clockwise edge, or no move is denoted by $1$,$-1$,$0$ respectively.}
      \label{fig:pingpongk}
\end{figure}


%

The Single-Group-Swiping Phase starts at round $r=0$ and all agents move clockwise in this phase, keeping uniform distribution. The phase ends at the first round $r'$ when an agent is blocked. Starting from round $r'$, the Two-Groups-Swiping phase starts.
Recall that $BC_{r'}$ is the set of agents trying to move clockwise in round $r'$ that encounter a missing edge. Since the agents are in distinct nodes, only one agent, say agent $a_j \in BC_{r'}$. This breaks the symmetry among the agents and they can partition themselves in two groups: group clockwise $G_{C}$ and group counter-clockwise $G_{CC}$. The group $G_{C}$ contains agent $a_{(j+2t) \mod k}$ with $t \in \mathbb{N}$, and group $G_{CC}$ contains all other agents (see Figure \ref{fig:problem1}). The partition into groups happens during the computation phase of round $r'$. 
From round $r'$, the agents move according to the following rules:
\begin{itemize}
\item Rule 1 (Group Movement): For $X \in {C,CC}$, an agent in $G_{X}$ moves in direction $X$ if no agent in $G_{X}$ is blocked, i.e. $\nexists a \in BX_{r} \bigcap G_{X}$ . This predicate is represented by the loops in states $C$, $CC$ of Figure \ref{fig:pingpongk}. 
\item Rule 2 (Membership Swapping): If at some round $r''$ agents in both groups are blocked, then the agents in $BC_{r''}$ and $BCC_{r''}$ swap their role, i.e. they exchange their states and thus their group membership in this round. Any other agent in $G_{X}$ moves in direction $X$ during this round. This rule is represented by the arrows that connect state $C$ and $CC$ for the blocked agents in Figure \ref{fig:pingpongk}. 
\end{itemize}

Intuitively, for Rule $1$ a group $G_{X}$ moves when all the agents in the group would be able to move without trying to cross a missing edge. Rule $2$ is applied only when two agents, one from group $G_{C}$ and one from group $G_{CC}$ are on two nodes that share the same missing edge, and this allows the groups to perform a  ``{\em virtual movement}'', see Figures \ref{fig:problem2}-\ref{fig:problem3}. 
 
 \begin{figure}[H]
\centering
\begin{subfigure}{0.30\textwidth}
\centering
    \includegraphics[width=\linewidth]{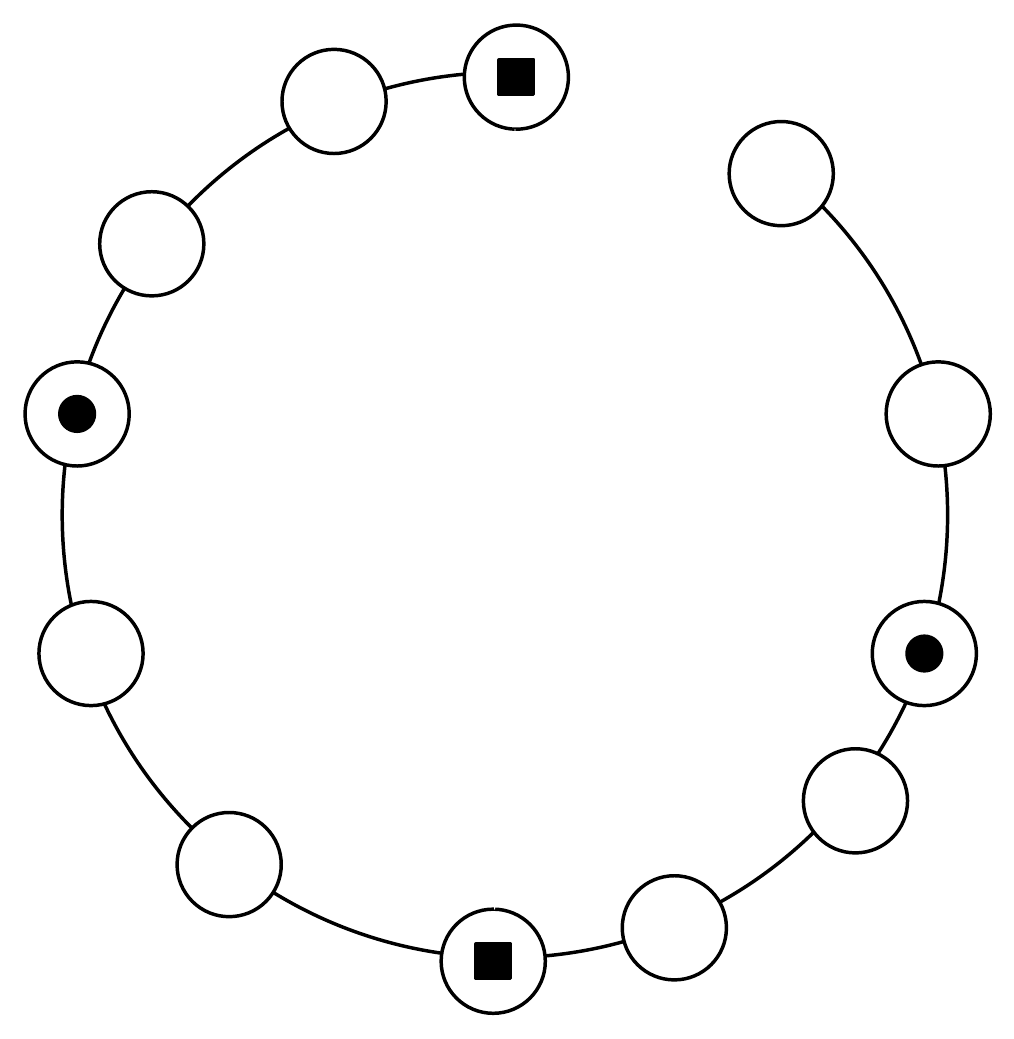}
    \caption{Starting round of Two-Groups-Swiping. $G_{C}$ (resp. $G_{CC}$) agents are marked with squares (dots)}
    \label{fig:problem1}
\end{subfigure}%
\quad
\begin{subfigure}{0.30\textwidth}
\centering
    \includegraphics[width=\linewidth]{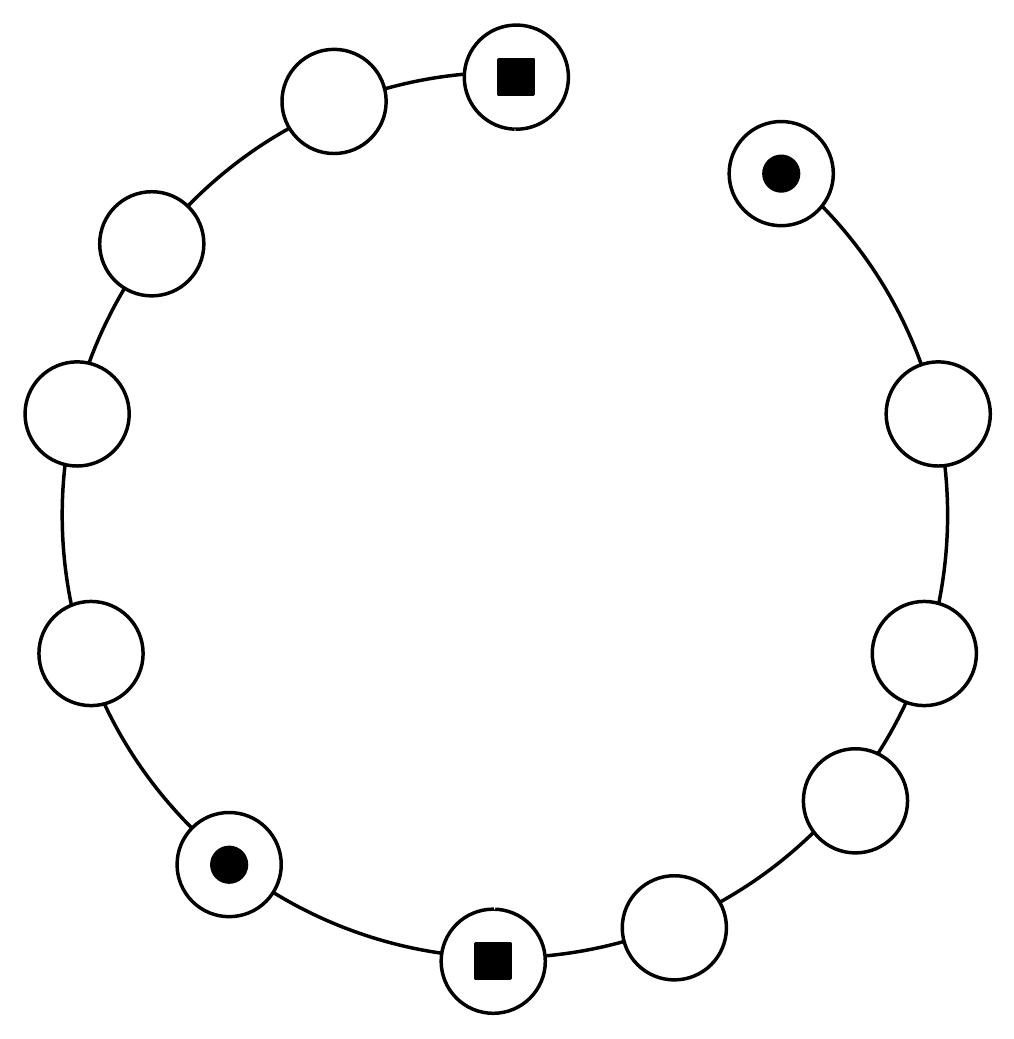}
    \caption{ Rule 1: Group $G_{C}$ is blocked. $G_{CC}$ reaches the other endpoint of the missing edge.}
    \label{fig:problem2}
\end{subfigure}
\quad
\begin{subfigure}{0.30\textwidth}
\centering
    \includegraphics[width=\linewidth]{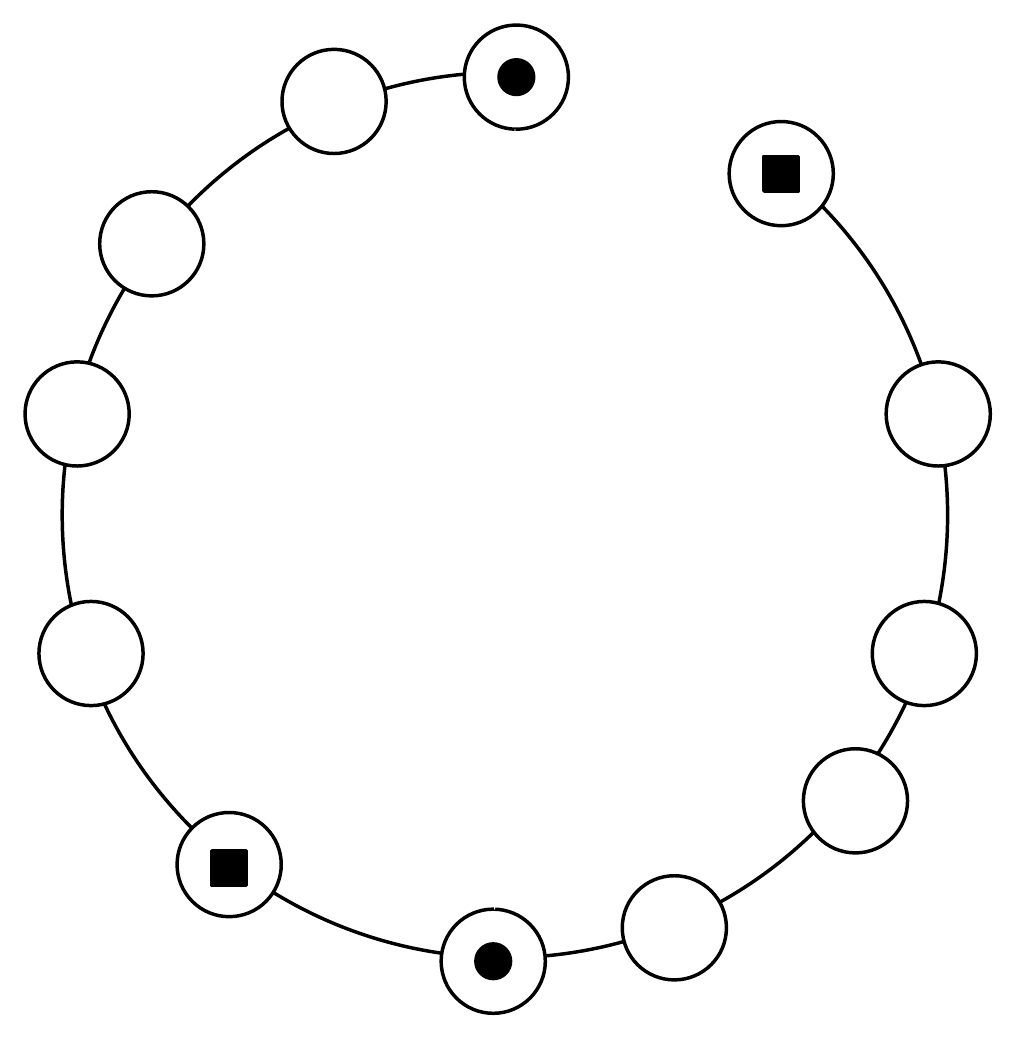}
    \caption{ Rule 2: The two blocked agents swap roles. Others move normally.}
    \label{fig:problem3}
\end{subfigure}
\caption{Algorithm \pingpongk, depiction of salient cases.}
\label{fig:problem}
\end{figure}

\begin{theorem}\label{twogroups:unkown}
The \pingpongk algorithm has an idle time of $\frac{4n}{k}$.
\end{theorem}
\begin{proof}
Consider two successive visits of a node $v$ at $r_0$ and $r_1$. 
If $r_0$ and $r_1$ belongs to the Single-Group-Swiping phase, then $r_1-r_0 \leq \frac{n}{k}$: the distance between two agents at most $\frac{n}{k}$ and they are all moving in the same direction.

Let us examine now the case when $r_0$ and $r_1$ belongs to the Two-Groups-Swiping phase. Notice that: (1) in each round at least one group moves, and (2) At round $r_0+1$, the distances between node $v$ and the closest agents in $G_{CC}$ (resp. $G_{CC}$) that are moving towards $v$ are at most $\frac{2n}{k}-1$.
Due to (1) we know that in the successive $\frac{4n}{k}$ rounds, at least one group performs $\frac{2n}{k}$ steps, thus reaching $v$.
Therefore $r_1-r_0 \leq \frac{4n}{k}$.

It remains to examine the last case when $r_0$ belongs to the Single-Group-Swiping phase and $r_1$ belongs to the Two-Groups-Swiping phase.
Thus we have $r_0 \leq r' \leq r_1$. Notice that, the worst case is the one in which $r' -r_0$ is maximised, that is $r' -r_0=\frac{n}{k}-1$: $v$ was about to be visited by an agent $a$, but $a$ switches direction.
In this case, the distance between $v$ and the first agent in $G_{CC}$ moving towards $v$, added to the distance between $v$ and the first agent in $G_{C}$ moving towards $v$, is $\frac{3n}{k}-1$. 
Using the same analysis of the previous case we obtain $r_1-r' \leq \frac{3n}{k}$ that gives $r_1-r_0 \leq \frac{4n}{k}$.
\end{proof}

\paragraph{When $k$ is not a divisor of $n$.}
In the case $k$ does not divide $n$, we have that in the initial placement the minimum distance between two agents is $\lfloor \frac{n}{k} \rfloor$ and the maximum distance is $\lfloor \frac{n}{k} \rfloor+1$. We can use the same analysis of Theorem \ref{twogroups:unkown}, taking into account the difference in the distance, which gives a bound of $\lfloor \frac{4n}{k} \rfloor+2$.

\paragraph{When agents are not uniformly placed.}
If agents are not uniformly placed initially, they can arrive at a uniform configuration in $O(n)$ steps.

\begin{observation}\label{obs:uniformspread}
Consider a set of $k \geq 2$ agents arbitrarily placed in a dynamic ring under the \unkwn model with global snapshot, then the agents need at most $2n$ rounds to reach an uniform placement in the ring.
\end{observation}

\begin{proof}
According to the initial configuration we may have that it is possible to find total order among agents or not.
Let us first assume the case when such total order does not exists. In this case the initial configuration is periodic, with period $P$, and it can be partitioned in $P$ segments.
In each of this segment a total order among robots exists, thanks to the presence of chirality. Therefore in each segment we can uniquely identify an agent and place it in the correct position. 
If there are no edge removal this terminates in at most $n$ rounds. If there are edge removal then we will show that there exists a total order. So we will use the algorithm for the total order case.

Let us consider the case when there is a total order, or when an edge is remove.
When an edge is removed the symmetry between agents is broke: the first agent in the total order, let it be $a_0$, is the nearest clockwise agent to the removed edge. The others are ordered according to the clockwise direction starting from $a_0$.
The agents place themselves uniformly using agent $a_0$ as reference.
Agent $a_{j}$, with $j >0$ moves clockwise or counter-clockwise according to the need of shrinking or expanding its distance from the final position of $a_{j-1}$. In case $a_j$ is blocked, then all agents $a_{i}$ with $i < j$ move counter-clockwise or clockwise to set $a_j$ in the correct position. This process requires at most $n$ rounds. 
The bound of $2n$ rounds comes from composing the two previous algorithm in the obvious way: if there is no total order we run the algorithm for periodic configurations and we switch to the one with total order as soon as the adversary introduces a failure. 
\end{proof}

\paragraph{When $k$ is odd.}
The problem for odd $k$ is that once the algorithm switches to the Two-Group-Swiping phase, the groups $G_{C},G_{CC}$ do not have equal sizes. One group has size $\frac{k-1}{2}$ and the other $\frac{k+1}{2}$.
Moreover, within each group the members are not uniformly placed. The last problem is easy fixable at the price of stabilization time using Observation \ref{obs:uniformspread}.
Once the groups are uniformly placed, we can bound the idle time to $\frac{4nk}{k^2-1}+4$, as shown in the following lemma:
\begin{lemma}
When one group has size $\frac{k-1}{2}$ and the other  $\frac{k+1}{2}$, the Two-Groups-Swiping phase of \pingpongk has an idle time of at most $\frac{4nk}{k^2-1}+4$ rounds.
\end{lemma}
\begin{proof}
W.l.o.g. let $G_{C}$ be the group of size $\frac{k-1}{2}$ and $G_{CC}$ be the other group. Let $r_0,r_1$ be the times between two successive visits of some node $v$. In the worst case at round $r_0+1$, node $v$ could be at distance at most $\frac{2n}{k-1}+1$ from an agent in group $G_{C}$, and distance at most $\frac{2n}{k+1}+1$ from an agent in  $G_{CC}$. The sum of these distances is $\frac{4nk}{k^2-1}+2$, and since only one group can be blocked at each round, this distance decreases by one at each round. This implies that $r_1-r_0 \leq \frac{4nk}{k^2-1}+2$, thus proving the bound. 
 \end{proof}

From the previous Lemma and using the same proof strategy of Theorem \ref{twogroups:unkown} we have that $\frac{4nk}{k^2-1}+4$ is the idle time of the algorithm. Unfortunately, it is not possible to bound the stabilization time of the algorithm.
The adversary decides when, and if, the algorithm goes to the Two-Groups-Swiping phase, and when this happen a certain number of rounds has to be payed to position in an uniform way the members of each group. However, in any infinite execution of the algorithm, there are only finitely many times in which two consecutive visits of a node are spaced by more than  $\frac{4nk}{k^2-1}+4$ rounds.

 \subsection{\kwn setting: \placeswipe for $k$ agents.}
 
Generalizing the algorithm Section \ref{sec:placeandswipe}, for $k$ agents is immediate. The algorithm is essentially the same, the only variations are: each phase now lasts $\lfloor \frac{n}{k} \rfloor$ rounds and $P_{i}$ is not a pair of nodes but $k$ nodes uniformly placed. Also in this case we assume that agents start uniformly placed, such assumption can be dropped sacrificing the stabilization time (see Observation \ref{obs:uniformspread}).
Lemma \ref{lemma:swipek} below is an equivalent of Lemma \ref{lemma:swipe} for $k\geq 2$ agents. Further, we can show that starting from any uniform configuration, the agents can reach, using the knowledge of ${\cal G}$, any given target uniform configuration in at most $\lceil \frac{n}{k} \rceil$ steps. 

\begin{lemma}\label{lemma:swipek}
Given any 1-interval connected dynamic ring ${\cal G}$, for any round $r_i$, 
there exists a set $P_i$ of $k$ uniformly spaced nodes, such that $k$ agents placed on $P_i$ and moving clockwise from round $r_i$ to round $r_{i}+ \lfloor \frac{n}{k} \rfloor$, together explore all nodes of the ring. 
\end{lemma}
\begin{proof}
The key idea to prove the existence of $P_{i}$ is Observation \ref{obs:ilcinkas}. By plugging $t=\lceil \frac{n}{k} \rceil -1$ in the statement of the observation. We have that there are $\lfloor \frac{n}{k} \rfloor +1$ nodes, let $E_i$ be this set, such that an agent being on one of these nodes at round
$r_i$ moving clockwise visits exactly $\lceil \frac{n}{k} \rceil$ nodes by the end of round $r_{i}+ \lfloor \frac{n}{k} \rfloor$. Now we have to prove that $E_i$ contains a set of uniformly placed nodes. But this is obvious since the ring contains at least $\lfloor \frac{n}{k} \rfloor$  uniformly placed nodes and the cardinality of $E_i$ is $\lfloor \frac{n}{k} \rfloor +1$.
Being the agents in $P_{i}$ uniformly placed, when each agent visits $\lceil \frac{n}{k} \rceil$ nodes the ring has been explored. 
\end{proof}

 \begin{theorem}\label{th:Kplaceswipe}
Consider the \knw model with global snapshots.
The algorithm \placeswipe allows $k$ agents with uniform initial placement to patrol a ring with an idle time $I(n) \leq  3\lceil \frac{n}{k} \rceil$.
\end{theorem}
\begin{proof}
The only thing to prove is that agents are able to reach $P_{i}$ during the $i$-th Placement Phase, since the correctness of the Swiping Phase is given by Lemma \ref{lemma:swipe}.
First notice that at round $start_{i}$ for each epoch $i$, agents are uniformly placed in the ring, since the agents are initially uniformly placed and at the end of each Swipe Phase they are still uniformly placed.  
We have to show that starting from any uniform configuration, the agents can reach, using the knowledge of ${\cal G}$, any target uniform configuration, within the time interval $[start_{i},start_{i}+ \lceil \frac{n}{k} \rceil -1]$.
At round $start_{i}$, the distance between agent $a_{j}$ and $a_{j-1}$ is at most $\lceil \frac{n}{k} \rceil$. Let $d_j$ be the distance between the point in $P_{i}$ that has the closest clockwise distance from $a_j$.
W.l.o.g $d_{j} \leq \lceil \frac{n}{k} \rceil-d_{j}$, otherwise we can switch to the counter-clockwise orientation instead. Note for two agents $a_l,a_j$ it could be that $|d_{j}-d_{l}| = 1$ (in case $k$ is not an exact divisor of $n$), w.l.o.g. let us consider that $d_{j}$ is max($d_{j},d_{l}$).
If each agent $a$ is able to move clockwise for at least $d_j$ steps we can reach configuration $P_i$. 
Otherwise, there exists at least one $a_{j}$ that has been blocked for at least $\lceil \frac{n}{k} \rceil-(d_{j}-1)$ rounds. This implies that by moving counter-clockwise each agent can move for at least
$\lceil \frac{n}{k} \rceil-d_{j}+1$ rounds reaching the target node in $P_{i}$.
 \end{proof}


\section{Conclusion}
We provided the first results on the patrolling problem in dynamic graphs. As patrolling is usually performed on boundaries of territories, it is natural to study the problem for ring networks. The results may be extended to other topologies e.g. by moving on any cycle containing all the nodes of a graph. Our results on the dynamic ring networks are almost complete, but there exists a small gap between the lower and upper bounds, specially for the case of $k>2$ agents which can be reduced by future work. In particular, we believe the lower bound for $k>2$ agents in the \unkwn setting can be improved.     


\bibliographystyle{plain} 
 \bibliography{mybibfile}
  
\end{document}